\documentclass [a4paper,fceqn,11pt]{article}
\usepackage{amsmath,amssymb,amsthm,latexsym}
\usepackage{graphics,graphicx}
\usepackage{color}

\newtheorem{theorem}{Theorem}[section]

\newtheorem{definition}{Definition}[section]
\newtheorem{remark}{Remark}[section]

\numberwithin{equation} {section}

\begin{document}
\begin{sloppy}
\title{\bf A Log Probability Weighted Moment Estimator of Extreme Quantiles}
%
\author{Frederico Caeiro
\\ {\small Universidade Nova de Lisboa, FCT and CMA}
\vspace{1.5pc}
\\
Dora Prata Gomes
\\ {\small Universidade Nova de Lisboa, FCT and CMA}
} 
\date{\today}
\maketitle

\small
\textbf{Abstract:}
In this paper we consider the semi-parametric estimation of extreme quantiles of a right heavy-tail model. We propose a new Log Probability Weighted Moment estimator for extreme quantiles, which is obtained from the estimators of the shape and scale parameters of the tail. Under a second-order regular variation condition on the tail, of the underlying distribution function, we deduce the non degenerate asymptotic behaviour of the estimators under study and present an asymptotic comparison at their optimal levels.  In addition, the performance of the estimators is illustrated through an application to real data.

\noindent
\normalsize

\section{Introduction}%
\label{caeirogomes_sec:1}%
Let us consider a set of $n$ independent and identically distributed (i.i.d.), or possibly weakly dependent and stationary random variables (r.v.’s), $X_1, X_2, \ldots, X_n$, with common distribution function (d.f.) $F$.  We shall assume that $\overline{F}:=1-F$ has a Pareto-type right tail, i.e., with the notation $g(x) \sim h(x)$ if and only if $g(x)/h(x) \rightarrow 1$,
as $x\rightarrow \infty$,
\begin{equation}
\label{caeirogomes_eq_paretotail}%
\overline{F}(x)\ \sim\  (x/C)^{-1/\gamma}, \qquad x\rightarrow \infty,
\end{equation}
with $\gamma>0$ and $C>0$ denoting the shape and scale parameters, respectively. 
Then the quantile function $U(t):=F^{\leftarrow}(1-1/t)=\inf\{x:F(x)\geq 1-1/t\}$,\quad $t>1$ is a regularly varying function with a positive index of regular variation equal to $\gamma$, i.e.,
\begin{equation}
\label{caeirogomes_firstorder}%
\lim_{t\rightarrow\infty}
\frac{U(tx)}{U(t)}
=x^{\gamma}\ .
\end{equation}
Consequentially, we are in the max-domain of attraction of the Extreme Value distribution
\begin{equation}
\label{caeirogomes_EV}%
EV_{\gamma}(x) :=\left\{
  \begin{array}{lll}
    \exp(-(1+\gamma x)^{-1/\gamma}),\: & 1+\gamma x > 0 \quad &\mbox{if} \quad \gamma \neq 0 \\
    \exp(-\exp(-x)), & x \in \mathcal{R} \quad &\mbox{if} \quad \gamma =0.
  \end{array}
 \right.
\end{equation}
and denote this by $F \in \mathcal{D}_M (EV_\gamma)$. The parameter $\gamma$  is called the extreme value  index (EVI), the primary parameter in Statistics of Extremes.

Suppose that we are interested in the estimation of a extreme quantile $q_{p}$, a extreme value exceeded with probability  $p=p_n\rightarrow 0$, small. Since $q_{p}=F^{\leftarrow}(1-p)\ \sim\  Cp^{-\gamma}$,\quad $p\rightarrow 0$, for any heavy tailed model under \eqref{caeirogomes_eq_paretotail}, we will also need to deal with the estimation of the shape and scale parameters $\gamma$ and $C$, respectively. Let $X_{n-k:n}\leq \ldots \leq X_{n-1:n}\leq X_{n:n}$ denote the sample of the $k+1$ largest order statistics (o.s.) of the sample of size $n$, where $X_{n-k:n}$ is a intermediate o.s., i.e., $k$ is a sequence of integers between $1$ and $n$ such that
\begin{equation}
\label{caeirogomes_intermediate}%
k\rightarrow \infty \quad\mbox{and}\quad  k/n \rightarrow 0,\quad \mbox{as}\quad \ n \rightarrow \infty.
\end{equation}
The classic semi-parametric estimators of the parameters $\gamma$ and $C$, introduced in Hill (1975),  
are 
\begin{equation}
\label{caeirogomes_hill}%
\hat\gamma^H_{k,n} :=\frac{1}{k}\sum_{i=1}^k \left(\ln
X_{n-i+1:n}- \ln
X_{n-k:n}\right),
\quad k=1,2,\ldots,n-1,
\end{equation}
and
\begin{equation}
\label{caeirogomes_C}%
\hat C^{H}_{k,n} :=X_{n-k:n}\ \left(\frac{k}{n}\right)^{\hat\gamma^H_{k,n}} ,
\quad k=1,2,\ldots,n-1,
\end{equation}
respectively. The EVI estimator in \eqref{caeirogomes_hill} is the well know Hill estimator, the average of the log excesses over the high threshold $X_{n-k:n}$.
The classic semi-parametric extreme quantile estimator is the Weissman-Hill estimator (Weissman, 1978) with functional expression
\begin{equation}
\label{caeirogomes_estimador_wh}%
 \hat W^{^{H}}_{k,n}(p)  := X_{n-k:n}\ \Big(\frac{k}{np}\Big)^{\hat\gamma^H_{k,n}}, \quad k=1,2,\ldots,n-1.
\end{equation}

Most classical semi-parametric estimators of parameters of the right tail usually exhibit the same type of behaviour, illustrated in Figure \ref{caeirogomes_fig_biasvariance}: we have a high variance for high thresholds $X_{n-k:n}$, i.e., for small values of $k$ and high bias for low thresholds, i.e., for large values of $k$. Consequently, the mean squared error (MSE) has a very peaked pattern, making it difficult to determine the optimal $k$, defined as the value $k_0$ where the MSE is minimal. For a detailed review on the subject see for instance Gomes \emph{et al.} (2008) and Beirlant \emph{et al.} (2012).
\begin{figure}[h!]
\begin{center}
\includegraphics[width=.5\textwidth]{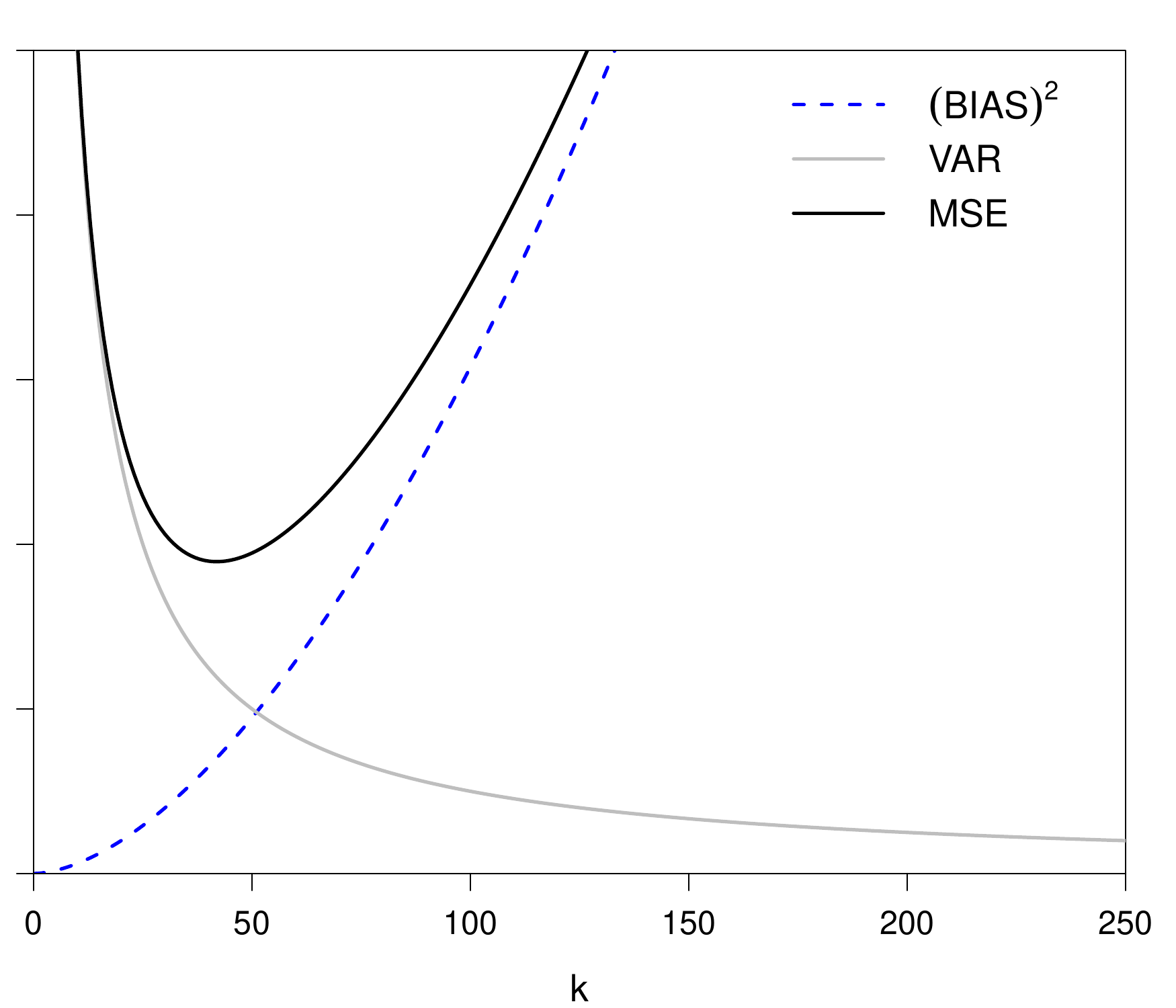}%
\end{center}
%
%
\caption{
Illustration of the Asymptotic Squared Bias, Variance and Mean squared error patterns, as function of $k$, of most classical semi-parametric estimators, for a sample of size $n = 250$.\vspace{10pt}}
\label{caeirogomes_fig_biasvariance} %
\end{figure}

Apart from the classical EVI, scale and extreme quantile estimators in \eqref{caeirogomes_hill}, (\ref{caeirogomes_C}) and (\ref{caeirogomes_estimador_wh}), respectively, we shall introduce in Section 2 the corresponding Log Pareto Probability Weighted Moment estimators. In Section \ref{caeirogomes_sec3}, we derive their non degenerate asymptotic behaviour and present an asymptotic comparison of the estimators under study at their optimal levels.


\medskip%
\section{Pareto Log Probability Weighted Moment Estimators}
The probability weighted moments (PWM) method, introduced in Greenwood \textit{et al.} (1979) 
 is a generalization of the method of moments. The PWM of a r.v. $X$,  are defined by $M_{p,r,s}:= E(X^p(F(X))^r(1-F(X))^s)$,
with $p,\ r,\ s\ \in \mathcal{R}$.  When $r=s=0$, $M_{p,0,0}$ are the usual non-central moments of order $p$. Hosking \textit{et al.} (1987) 
advise the use of  $M_{1,r,s}$ because the relation between parameters and moments is usually simpler than for the non-central moments. Also, if $r$ and $s$ are positive integers, $F^r(1-F)^s$ can be written as a linear combination of powers of $F$ or $1-F$ and usually work with one of the moments $a_r:=M_{1,0,r}=E(X(1-F(X))^r)$ or $b_r:=M_{1,r,0}=E(X(F(X))^r)$. Given a sample size $n$, the unbiased estimators of $a_r$ and $b_r$ are, respectively,
\[
\hat a_r=\frac{1}{n}\sum_{i=1}^{n-r}\frac{\binom{n-i}{r}}{\binom{n-1}{r}}X_{i:n},\quad
\text{and}\quad
\hat b_r=\frac{1}{n}\sum_{i=r+1}^n\frac{\binom{i-1}{r}}{\binom{n-1}{r}}X_{i:n}.
\]
The first semi-parametric Pareto PWM (PPWM) estimators for heavy tailed models appeared in Caeiro and Gomes (2011a), for the estimation of the shape and scale parameters $\gamma$ and $C$, and in Caeiro \emph{et al.} (2012), for the estimation of extreme quantiles and tail probabilities. Since all those PPWM estimators use the sample mean, they are only consistent if $0<\gamma<1$. Caeiro and Gomes (2013) generalized the estimators in Caeiro and Gomes (2011a) with a class of PPWM estimators, consistent for\, $0<\gamma<1/r$\, with\, $r>0$. In order to remove the right-bounded support of the previous PPWM estimators and have consistent estimators for every $\gamma>0$, we shall next introduce new semi-parametric estimators based on the log-moments
$$l_{r}:= E((\ln X)(1-F(X))^r).$$
For non-negative integer $r$, the unbiased estimator of $l_{r}$ is given by
\[
\hat l_r=\frac{1}{n}\sum_{i=1}^{n-r}\frac{\binom{n-i}{r}}{\binom{n-1}{r}}\ln X_{i:n}.
\]
 For the strict Pareto model with d.f. $F(x)=1-(x/C)^{-1/\gamma}$, $x>C>0$, $\gamma>0$ the Pareto log PWM (PLPWM) are\ $l_{r}={\ln (C)}/{(1+r)}+{\gamma}/{(1+r)^2}$.
 
To obtain the tail parameters estimators of $\gamma$ and $C$ of a underlying model with d.f. under \eqref{caeirogomes_eq_paretotail}, we need the followings results:
\begin{itemize}
\item $\frac{X_{n-k:n}}{C(n/k)^\gamma}$ converges in probability to 1, 
for intermediate $k$;
\item the conditional distribution of $X|X>X_{n-k:n}$, is approximately Pareto with shape parameter $\gamma$ and scale parameter $C(n/k)^\gamma$.
\end{itemize} 
The PLPWM estimators of $\gamma$ and $C$, based on the $k$ largest observations, are
\begin{equation}
\label{caeirogomes_eq_PLPWM_EVI}%
\hat \gamma^{^{PLPWM}}_{k,n} :=\frac{1}{k}\sum_{i=1}^k \left(2-4\frac{i-1}{k-1}\right)\ln
X_{n-i+1:n}
,\quad k=2,\ldots,n,
\end{equation}
and
\begin{equation}
\hat C^{^{PLPWM}}_{k,n}:=\Big(\frac{k}{n}\Big)^{\hat\gamma^{^{PLPWM}}_{k,n}}\exp\left\{D_{k,n}\right\},\quad k=2,\ldots,n,
\end{equation}
with  $D_{k,n}:=\frac{1}{k}\sum_{i=1}^k \left(4\frac{i-1}{k-1}-1\right)\ln
X_{n-i+1:n}$. Notice that $\hat \gamma^{^{PLPWM}}_{k,n}$ is a weighted average of the $k$ largest observations, with the weights $g_{i,k}:=(2-4\frac{i-1}{k-1})$. Since $g_{i,k}=-g_{k-i+1,k}$, the weights are antisymmetric and their sum is zero.  
On the basis of the limit relation $q_{p} \sim  Cp^{-\gamma}$,\, $p\rightarrow 0$, we shall also consider the following quantile estimator
\begin{equation}
\label{caeirogomes_estimador_qPLPWM}%
\hat Q^{^{PLPWM}}_{k,n}(p):=\Big(\frac{k}{np}\Big)^{\hat\gamma^{^{PLPWM}}_{k,n}}\exp\left\{D_{k,n}\right\},\quad k=2,\ldots,n,
\end{equation}
valid for $\gamma>0$.

\medskip%
\section{Asymptotic Results}%
\label{caeirogomes_sec3}%
\subsection{Non Degenerate Limiting Distribution}%
In this section we derive several basic asymptotic results for the EVI estimators in \eqref{caeirogomes_hill} and \eqref{caeirogomes_eq_PLPWM_EVI} and for the quantiles estimators, $\hat W^{^{H}}_{k,n}(p)$ and $\hat Q^{^{PLPWM}}_{k,n}(p)$. Asymptotic results for the scale $C$-estimators are not presented but can be obtained with an analogous proof.

To ensure the consistency of the EVI semi-parametric estimators, for all $\gamma>0$, we need to assume that $k$ is an intermediate sequence of integers, verifying \eqref{caeirogomes_intermediate}.
To study the asymptotic behaviour of the estimators, we need a second order regular variation 
condition with a parameter $\rho\leq 0$ that measures the rate of convergence of  $U(tx)/U(t)$ to $x^\gamma$ in \eqref{caeirogomes_firstorder} and is given by
\begin{equation}
\label{caeirogomes_secondorder}%
\lim_{t\rightarrow\infty} \frac{\ln U(tx)-\ln U(t)-\gamma\ln x}{A(t)}=\frac{x^{\rho}-1}{\rho}
\Leftrightarrow\ \lim_{t\rightarrow\infty}\frac{\frac{U(tx)}{U(t)}-x^\gamma}{A(t)}=x^\gamma \frac{x^\rho-1}{\rho},
\end{equation}
for all $x > 0$, with $|A|$ a regular varying function with index $\rho$ and $\frac{x^\rho-1}{\rho}=\ln x$ if $\rho=0$. 
\begin{theorem}
\label{caeiro_gomes_t1}%
Under the second order framework, in (\ref{caeirogomes_secondorder}), and for intermediate $k$,  i.e, whenever (\ref{caeirogomes_intermediate}) holds, the asymptotic distributional representation of   $\hat\gamma^\bullet_{k,n}$, with $\bullet$  denoting either $H$ or $PLPWM$, is given by
\begin{equation}
\label{caeirogomes_eq_EVI_dist}%
\hat\gamma^\bullet_{k,n}\ \stackrel{d}{=}\ \gamma+\frac{\sigma_{_\bullet} Z_k^\bullet}{\sqrt{k}}+ b_{_\bullet} {A(n/k)}(1+o_p(1)),
\end{equation}
where $\overset{d}{=}$\ denotes equality in distribution, $Z_k^\bullet$ is a standard normal r.v.,
\[
b_{_H}=\frac{1}{1-\rho},\quad b_{_{PLPWM}}=\frac{2}{(1-\rho)(2-\rho)},\quad \sigma_{_H}=\gamma\quad and\quad \sigma_{_{PLPWM}}=\frac{2}{\sqrt{3}}\gamma.
\]
If we choose the intermediate level $k$ such that $\sqrt{k}\, A(n/k) \rightarrow\lambda\in \mathcal{R}$, then, 
\[
\sqrt{k}(\hat\gamma^\bullet_{k,n}-\gamma)\overset{d}{\rightarrow}N(\lambda\, b_{_\bullet}, \sigma_{_\bullet}^2).
\] 
\end{theorem}
\begin{proof}
For the Hill estimator, the proof can be found in de Haan and Peng (1998). 
For the PLPWM EVI-estimator, note that $\sum_{i=1}^k \left(2-4\frac{i-1}{k-1}\right)=0$ and consequently 
\[
\hat \gamma^{^{PLPWM}}_{k,n} =\frac{1}{k}\sum_{i=1}^k \left(2-4\frac{i-1}{k-1}\right)\ln\frac{
X_{n-i+1:n}}{X_{n-k:n}}=%
\frac{1}{k}\sum_{i=1}^k g_{i,k}\ln\frac{
X_{n-i+1:n}}{X_{n-k:n}}, \quad k<n.
\]
We can write $X\overset{d}{=}U(Y)$ where $Y$ is a standard Pareto r.v., with d.f. $F_Y(y)=1-1/y$, $y>1$. Consequently and provided that $k$ is intermediate, we can apply equation \eqref{caeirogomes_secondorder}  with $t=Y_{n-k:n}$ and $x=Y_{n-i+1:n}/Y_{n-k:n}\overset{d}{=}Y_{k-i+1:k}$, $1\leq i\leq k$, to obtain
\[
\ln\frac{X_{n-i+1:n}}{X_{n-k:n}}%
\stackrel{d}{=}%
\gamma\ln Y_{k-i+1:k}+\frac{Y_{k-i+1:k}^\rho-1}{\rho}A(Y_{n-k:n})(1+o_p(1)).
\]
Then, since $nY_{n-k:n}/k\overset{p}{\rightarrow} 1$, as $n\rightarrow\infty$,
\[
\hat \gamma^{^{PLPWM}}_{k,n}\overset{d}{=}\frac{1}{k}\sum_{i=1}^k g_{i,k}
\left\{\gamma E_{k-i+1:k}+\frac{Y_{k-i+1:k}^\rho-1}{\rho}A(n/k)(1+o_p(1))\right\},
\] 
where $\{E_i\}_{i\geq 1}$, denotes a sequence of i.i.d. standard exponential r.v.'s. The distributional representation of the EVI-estimator $\hat\gamma^{^{PLPWM}}_{k,n}$ follows  from the results for linear functions of ordinal statistics (David and Nagaraja, 2003), i.e.,  
$
Z_k^{^{PLPWM}}=\frac{\sqrt{k}}{\sigma_{PLPWM}}\frac{1}{k}\sum_{i=1}^k (g_{i,k}E_{k-i+1:k}-1)
$
is a standard normal r.v. and
$\frac{1}{k}\sum_{i=1}^k g_{i,k}
\frac{Y_{k-i+1:k}^\rho-1}{\rho}$ converges in probability towards $\frac{2}{(1-\rho)(2-\rho)}$, as $k \rightarrow \infty$.\\
The asymptotic normality of $\sqrt{k}(\hat\gamma^\bullet_{k,n}-\gamma)$ follows straightforward from the  representation in distribution in \eqref{caeirogomes_eq_EVI_dist}.
\end{proof}

\begin{remark}
Notice that $\hat \gamma^{^{PLPWM}}_{k,n}$ has a smaller asymptotic bias, but a larger asymptotic variance than $\hat \gamma^{^{H}}_{k,n}$. A more precise comparison of the EVI-estimators will be dealt in Section \ref{caeirogomes_sec_asymp_comp}.
\end{remark}

\begin{remark}
For intermediate $k$  such that $\sqrt k\, A(n/k) \to \lambda$, finite, as $n \to \infty$, the Asymptotic Mean Squared Error (AMSE) of any semi-parametric EVI-estimator, with asymptotic distributional representation given by \eqref{caeirogomes_eq_EVI_dist}, is
$$
AMSE(\widehat{\gamma}_{n, k}^{\bullet}) := \frac{\sigma_{\bullet}^2}{k} + b_{\bullet}^2 A^2(n/k),
$$
where $Bias_{\infty}(\widehat{\gamma}_{n, k}^{\bullet}) :=b_{\bullet} A(n/k)$ and $Var_{\infty}(\widehat{\gamma}_{n, k}^{\bullet}) :=\sigma_{\bullet}^2 /k$. Let $k_0^{\bullet}$ denote the level $k$, such that $AMSE(\widehat{\gamma}_{n, k}^{\bullet})$ is minimal, i.e., $k_0^{\bullet}\equiv k_0^{\bullet}(n):=\arg \min_k AMSE(\widehat{\gamma}_{n, k}^{\bullet})$.
If $A(t)=\gamma\beta t^\rho$, \, $\beta\neq 0$, $\rho<0$ which holds for most common heavy tailed models, like the Fr\'{e}chet, Burr, Generalized Pareto or Student's t, the optimal $k$-value for the EVI-estimation through $\hat\gamma_{n, k}^{\bullet}$ is well approximated by
\begin{equation}
\label{caeirogomes_k0}%
k_0^{\bullet}=\left(\frac{\sigma_{\bullet}^2n^{-2\rho}}{(-2\rho)b_{\bullet}^2\gamma^2\beta^2}\right)^{\frac{1}{1-2\rho}}.
\end{equation}
\end{remark}
\begin{remark}
The estimation of the shape second-order parameter $\rho$ can be done using the classes of estimators in Fraga Alves \emph{et al.} (2003), Ciuperca and Mercadier (2010), Goegebeur \emph{et al.} (2010) or Caeiro and Gomes (2012). Consistency of those estimators is achieved for intermediate $k$ such that 
$\sqrt{k}A(n/k)\rightarrow \infty$ as $n\rightarrow\infty$. For the estimation of the scale second-order parameter $\beta$, for models with $A(t)=\gamma\beta t^\rho$, \, $\beta\neq 0$, $\rho<0$, we refer the reader to the estimator in Gomes and Martins (2002). That estimator is consistent for intermediate  $k$ such that $\sqrt{k}A(n/k)\rightarrow \infty$ as $n\rightarrow\infty$ and estimators of $\rho$ such that $\hat\rho-\rho=o_p(1/\ln n)$. Further details on the estimation of ($\rho$,$\beta$) can be found in Caeiro \emph{et al.} (2009).
\end{remark}

For the extreme quantile estimators in \eqref{caeirogomes_estimador_wh} and \eqref{caeirogomes_estimador_qPLPWM}, their asymptotic distributional representations follows from the next, more general, Theorem.
\begin{theorem}
Suppose that 
$\bullet$ denotes any EVI-estimator with distributional representation given by \eqref{caeirogomes_eq_EVI_dist}. Under the conditions of Theorem  \ref{caeiro_gomes_t1}, if  $p=p_n$ is a sequence of probabilities such that $c_n:=k/(np)\ \underset{}{\rightarrow}\ \infty$, $\ln c_n=o(\sqrt{k})$ and $\sqrt{k}A(n/k)\rightarrow \lambda \in \mathcal{R}$, as $n\rightarrow\infty$, then, 
\begin{equation}
\label{caeirogomes_quantil1}%
\frac{\sqrt{k}}{\ln c_n}\left(\frac{\hat Q^{\bullet}_{k,n}(p)}{q_{p}}-1\right)\overset{d}{=}
\frac{\sqrt{k}}{\ln c_n}\left(\frac{\hat W^{\bullet}_{k,n}(p)}{q_{p}}-1\right)\overset{d}{=}
\sqrt{k}\left(\hat \gamma^{\bullet}_{k,n}-\gamma\right)
(1+o_p(1)).
\end{equation}
\end{theorem}
\begin{proof}
Since  $q_{p}=U(1/p)$, we can write
\[\frac{\hat W^{\bullet}_{k,n}(p)}{q_{p}}=\frac{X_{n-k:n}}{U(n/k)}.\frac{U(n/k)}{U(nc_n/k)}(c_n)^{\hat \gamma^{\bullet}_{k,n}}.\]
Using the second order framework, in \eqref{caeirogomes_secondorder}, with $t=n/k$ and $x=\frac{k}{n}Y_{n-k:n}$,  results in
$\frac{X_{n-k:n}}{U(n/k)}\overset{d}{=}1+\frac{\gamma}{\sqrt{k}}B_k+o_p(A(n/k))$
where $B_k:=\sqrt{k}\left(\tfrac{k}{n} Y_{n-k:n}-1\right)$ is asymptotically a standard normal random variable. Using the results in de Haan and Ferreira (2006), Remark B.3.15 (p. 397),  
$\left(\frac{U(c_n.n/k)}{U(n/k)c_n^\gamma}\right)^{-1}=1+\frac{A(n/k)}{\rho}(1+o(1))$ follows.  Then, since $(c_n)^{\hat \gamma^{\bullet}_{k,n}-\gamma}\overset{d}{=}1+\ln (c_n)(\hat \gamma^{\bullet}_{k,n}-\gamma)(1+o_p(1))$, we get
\begin{equation*}
\frac{\hat W^{\bullet}_{k,n}(p)}{q_{p}}\overset{d}{=}1+\ln (c_n)(\hat \gamma^{\bullet}_{k,n}-\gamma)(1+o_p(1))
+\frac{\gamma B_k}{\sqrt{k}}+\frac{A(n/k)}{\rho}(1+o_p(1)),
\end{equation*}
and the second equality in \eqref{caeirogomes_quantil1}  follows immediately.

For the other quantile estimator, we can write
\[
\hat Q^{\bullet}_{k,n}(p)=X_{n-k:n}\ \Big(\frac{k}{np}\Big)^{\hat\gamma^{\bullet}_{k,n}}  \exp\{\tilde D_{k,n}\}=\hat W^{\bullet}_{k,n}(p)\exp\{\tilde D_{k,n}\},
\]
with $\tilde D_{k,n}:=\frac{1}{k}\sum_{i=1}^k \left(4\frac{i-1}{k-1}-1\right)\ln\frac{
X_{n-i+1:n}}{X_{n-k:n}}$. Then, since we have 
$$\exp\{\tilde D_{k,n}\}\overset{d}{=}1+\frac{\gamma}{\sqrt{3k}}P_k-\frac{\rho A(n/k)(1+o_p(1))}{(1-\rho)(2-\rho)},$$
with $P_k$ a standard normal r.v., the first equality in \eqref{caeirogomes_quantil1}  follows.
\end{proof}

\medskip%
\subsection{Asymptotic Comparison at Optimal Levels}
\label{caeirogomes_sec_asymp_comp}%
We now proceed to an asymptotic comparison of the PLPWM EVI estimator in \eqref{caeirogomes_eq_PLPWM_EVI} with the Hill estimator in \eqref{caeirogomes_hill} and the PPWM EVI estimator in Caeiro and Gomes (2011a), at their optimal levels. This comparison is done along the lines of de Haan and Peng (1998), Gomes and Martins (2001), Caeiro and Gomes (2011b), among others. Similar results  hold  for the extreme quantile estimators, at their optimal levels, since they have the same asymptotic behaviour as the EVI
estimators, although with a slower convergence rate.

Let $k_0^{\bullet}$ be the optimal level for the estimation of $\gamma$ through $\widehat{\gamma}_{k,n}^{\bullet}$ given by \eqref{caeirogomes_k0}, i.e., the level associated with a minimum asymptotic mean square error, and let us denote $\widehat{\gamma}_{n0}^{\bullet}:=\widehat{\gamma}_{k_0^{\bullet},n }^{\bullet}$, the estimator computed at its optimal level.
Dekkers and de Haan (1993)  proved that, whenever $b_{\bullet} \neq 0$, there exists a function $\varphi(n; \gamma, \rho)$, dependent only on the underlying model, and not on the estimator, such that
\begin{equation}
\displaystyle{\lim_{n \to \infty}} %
\varphi(n; \gamma, \rho)AMSE(\widehat{\gamma}_{n0}^{\bullet}) = 
\left(\sigma_{\bullet}^2 \right)^{-\frac{2\rho}{1-2\rho}}\left(b_{\bullet}^2 \right)^{\frac{1}{1-2\rho}}=: LMSE(\widehat{\gamma}_{n0}^{\bullet}). 
\label{caeirogomes_LMSE}%
\end{equation}

\noindent
It is then sensible to consider the following:
\begin{definition}
Given  two biased estimators $\widehat\gamma_{n,k}^{(1)}$ and $\widehat\gamma_{n,k}^{(2)}$, for which distributional representations of the type  \eqref{caeirogomes_eq_EVI_dist}  hold with constants $(\sigma_1, b_1)$ and $(\sigma_2, b_2)$, $b_1$, $b_2 \neq 0$, respectively, both computed at their optimal levels, $k_0^{(1)}$ and $k_0^{(2)}$, the Asymptotic Root Efficiency  $(AREFF)$ indicator is defined as
\begin{equation}
\label{caeirogomes_areff}%
AREFF_{1|2}   := \sqrt{
{LMSE\left(\widehat{\gamma}_{n0}^{(2)}\right)}/
{ LMSE\left(\widehat{\gamma}_{n0}^{(1)}\right)}}=
\left(\left(\frac{\sigma_2}{\sigma_1}\right)^{-2\rho}
\left|\frac{b_2}{b_1}\right|\right)^{\frac{1}{1-2\rho}},
\end{equation}
with LMSE given in \eqref{caeirogomes_LMSE} and $\widehat{\gamma}_{n0}^{(i)}$ := $\widehat{\gamma}_{k_0^{(i)}, n }^{(i)}$,\ $i=1,2$.
\end{definition}

\begin{remark}
Note that this measure was devised so that the higher the AREFF indicator is, the better the first estimator is.
\end{remark}
\begin{remark}
For the PPWM EVI estimator, in Caeiro and Gomes (2011a), we have
\[
b_{_{PPWM}}\!=\frac{(1-\gamma)(2-\gamma)}{(1-\gamma-\rho)(2-\gamma-\rho)}\quad  \text{and}\quad \! \sigma_{_{PPWM}}\!=\frac{\gamma\sqrt{1-\gamma}(2-\gamma)}{\sqrt{1-2\gamma}\sqrt{3-2\gamma}},\quad 0<\gamma<0.5.
\]
\end{remark}

\medskip%
\noindent
To measure the performance of $\hat \gamma^{^{PLPWM}}_{k,n}$, we have computed the AREFF-indicator, in \eqref{caeirogomes_areff}, as function of the second order parameter $\rho$. In figure \ref{caeirogomes_fig_areff} (left), we present the values of
\begin{equation}
\label{caeirogomes_AREFF_PLPWM-H}%
AREFF_{PLPWM|H}(\rho)=
\left(\left(\frac{3}{4}\right)^{-\rho}
\left(1-\frac{\rho}{2}\right)\right)^{\frac{1}{1-2\rho}},
\end{equation}
as a function of $\rho$.
This indicator has a maximum near $\rho=-0.7$, and we have $AREFF_{PLPWM|H}>1$, if $-3.54<\rho<0$, an important region of $\rho$ values in practical applications. It is also easy to check that $\lim\limits_{\rho\rightarrow -\infty}AREFF_{PLPWM|H}(\rho)=\sqrt{3}/2\approx 0.866$ and $\lim\limits_{\rho\rightarrow 0}AREFF_{PLPWM|H}(\rho)=1$.\\ In figure \ref{caeirogomes_fig_areff} (right) we show a contour plot with the comparative behaviour, at optimal levels, of the PLPWM and PPWM EVI-estimators in an important region of the ($\gamma$,$\rho$)-plane. The grey colour marks the area where $AREFF_{PLPWM|PPWM}>1$. At optimal levels, there is only a small region of the ($\gamma$,$\rho$)-plane where the AREFF indicator is slightly smaller than 1. Also, the $AREFF_{PLPWM|PPWM}$ indicator increases, as $\gamma$ increases and/or $\rho$ decreases.

\begin{figure}[h!]
\center{
\includegraphics[width=.475\textwidth]{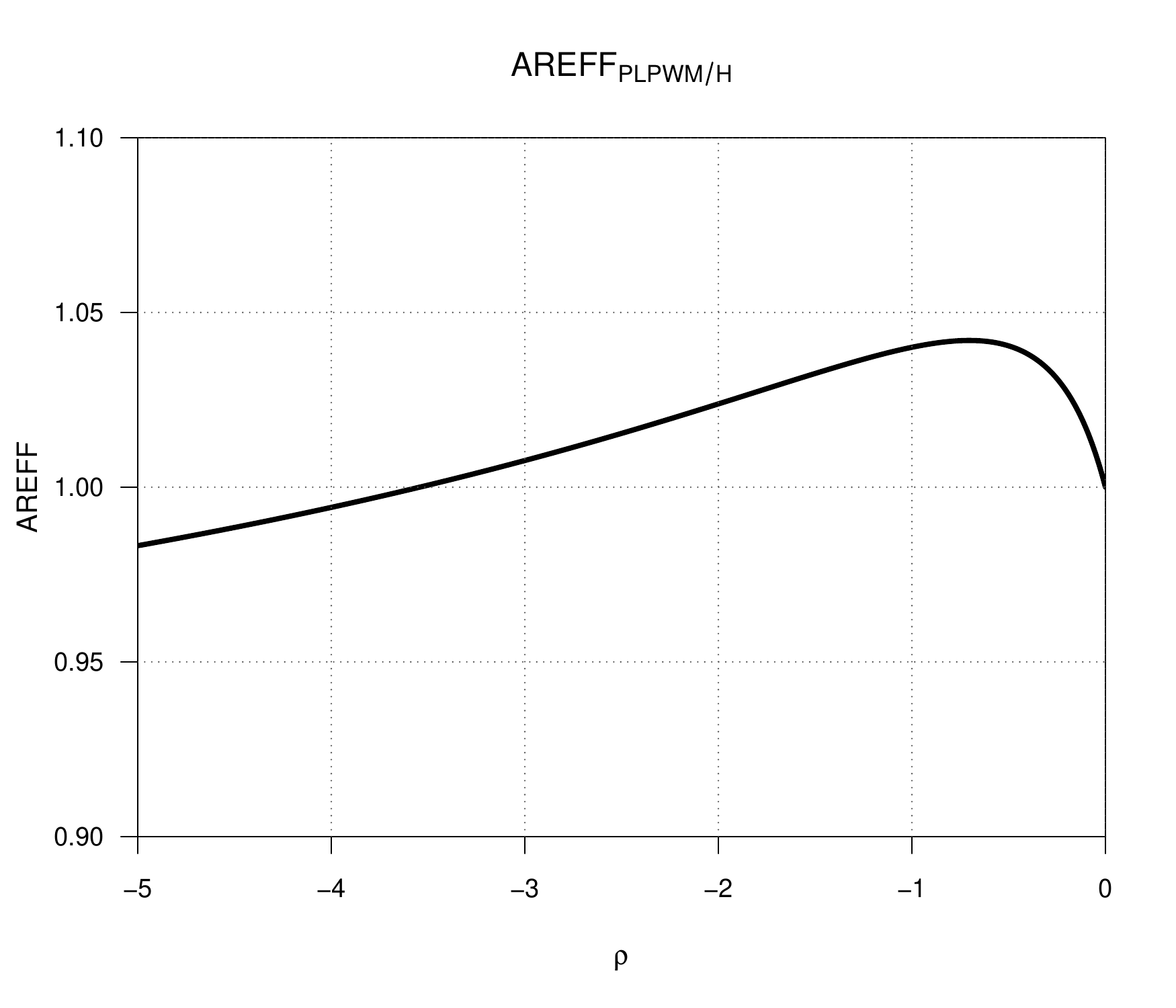}\quad
\includegraphics[width=.475\textwidth]{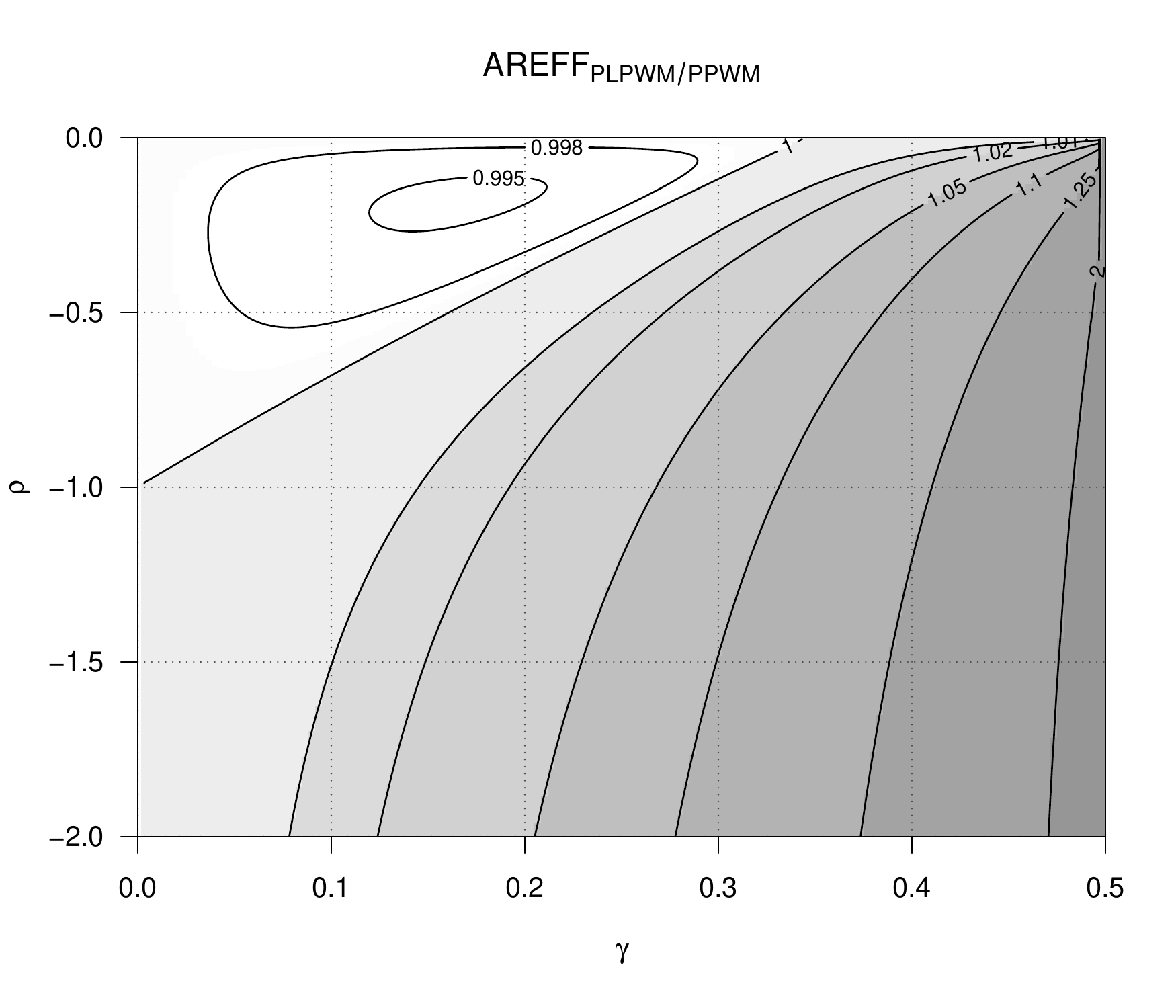}
}%
%
%
\caption{
\textbf{Left:} Plot with the indicator $AREFF_{PLPWM|H}(\rho)$, in \eqref{caeirogomes_AREFF_PLPWM-H}, as a function of $\rho$. \textbf{Right:} Contour plot with the indicator $AREFF_{PLPWM|PWM}$, as a function of ($\gamma$,$\rho$).
}
\label{caeirogomes_fig_areff} %
\end{figure}

\medskip%
\section{A Case Study}
As an  illustration of the performance of the estimators under study,
we shall next consider the analysis of the Secura Belgian Re automobile claim amounts exceeding 1,200,000 Euro, over the period 1988-2001. This data set of size $n=371$ was already studied by several authors  (Beirlant \emph{et al.}, 2004; Beirlant \emph{et al.}, 2008 and Caeiro and Gomes, 2011c). 

\vspace{-2pt}
\begin{figure}[h!]
\center{%
\includegraphics[width=.475\textwidth]{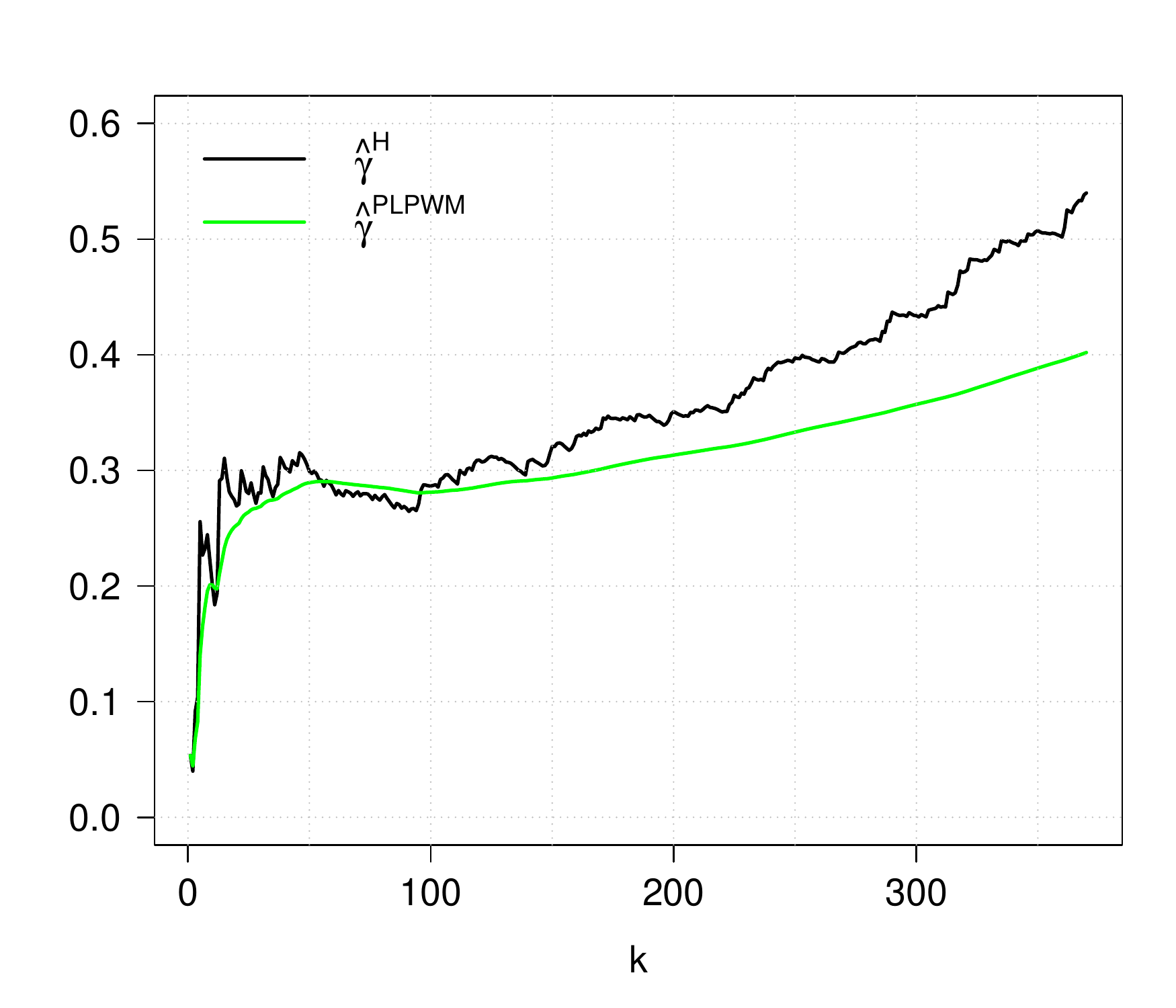}\quad
\includegraphics[width=.475\textwidth]{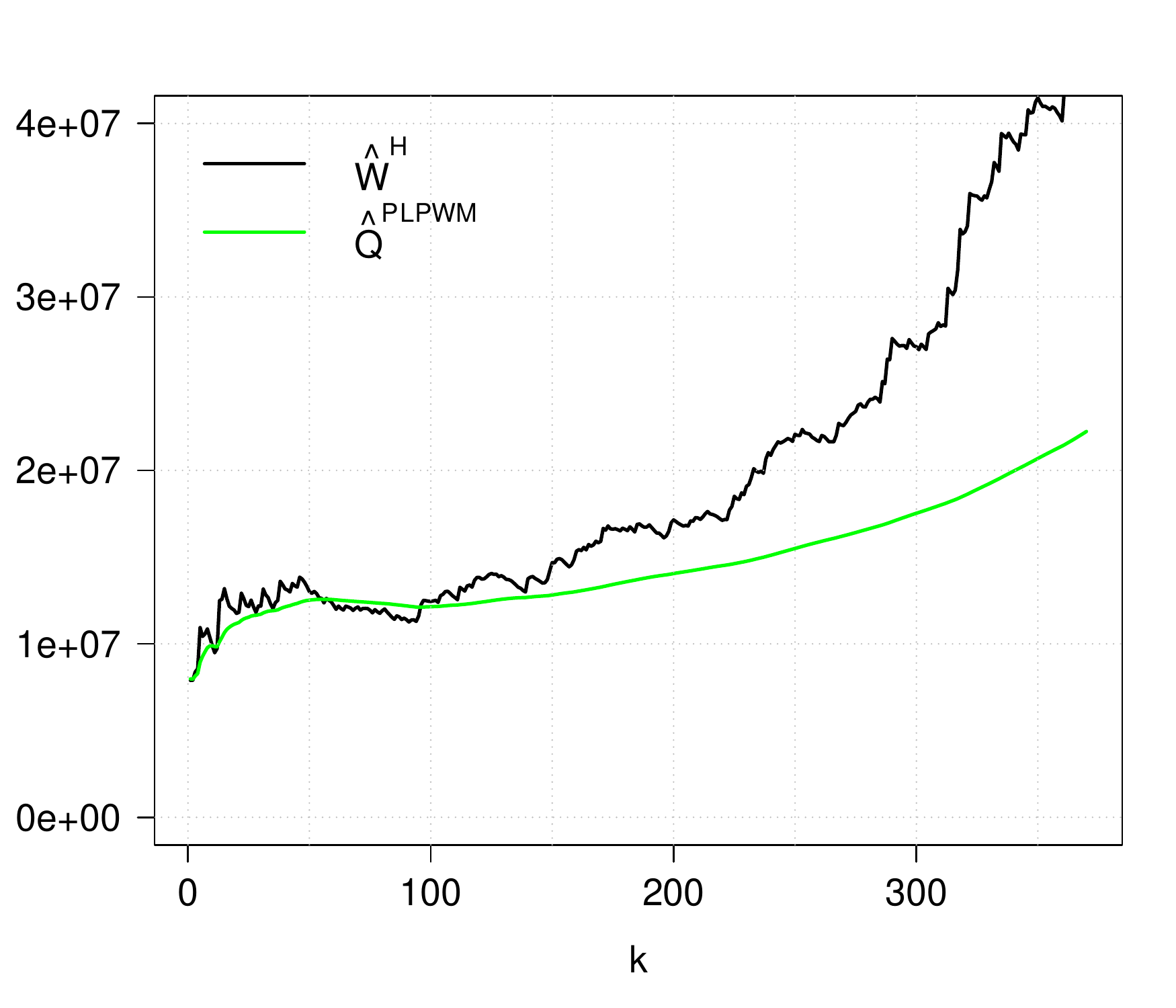}%
}%
%
%
\caption{
\textbf{Left:} Estimates of the EVI for the Secura
Belgian Re data; \textbf{Right:} Estimates of the quantile $q_p$ with $p = 0.001$ for the Secura Belgian Re data.
}
\label{caeirogomes_fig_secura} %
\end{figure}
In Figure \ref{caeirogomes_fig_secura}, we present, at the left, the EVI estimates provided by the Hill and PLPWM EVI-estimators in \eqref{caeirogomes_hill} and \eqref{caeirogomes_eq_PLPWM_EVI}, respectively. At the right we present the corresponding quantile estimates provided by Weissman-Hill and PLPWM estimators, in \eqref{caeirogomes_estimador_wh} and \eqref{caeirogomes_estimador_qPLPWM}, with $p = 0.001$. For a fair comparison of the PLPWM estimators with the equivalent classic estimators, the PLPWM estimators are now based on the top $k + 1$ largest o.s.'s. For this dataset, we have $\hat\rho=-0.756$ and $\hat\beta=0.803$, obtained at the level $k_1 = [n^{0.999}] = 368$ (Caeiro and Gomes, 2011c). Using these values, the estimates of the optimal level, given by \eqref{caeirogomes_k0}, are $\hat k_0^H=55$ and $\hat k_0^{PLPWM}=76$. Consequently, we have $\hat\gamma^{^{H}}_{55, 371}=0.291$ and $\hat\gamma^{^
{PLPWM}}_{76, 371}=0.286$. Finally, the quantile estimates are given by $\hat W^{H}_{55,371}(p)=12622248$ and $\hat Q^{PLPWM}_{76,371}(p)=12373324$.

%

\medskip%
\section{Some Overall Conclusions}
Based on the results here presented we can make the following comments:
\begin{itemize}
\item Regarding efficiency at optimal levels, the new PLPWM estimators are a valid alternative to the classic Hill, Weissman-Hill and PPWM estimators. And they are consistent for any $\gamma>0$, which does not happen for the PPWM estimators.

\item The analysis of the automobile claim amounts gave us the impression that the PLPWM EVI and extreme quantile estimators have a much smoother sample pattern than the Hill and the Weissman-Hill estimators.

\item It is also important to study the behaviour of the new PLPWM estimators for small sample sizes. That topic should be adressed in  future research work. \\
\end{itemize}

\medskip%

\noindent \textbf{acknowledgement:}\,
Research partially supported by 
FCT -- Funda\c{c}\~ao para a Ci\^{e}ncia e a Tecnologia, project 
PEst-OE/MAT/UI0297/2011 (CMA/UNL), EXTREMA, PTDC/MAT /101736/2008.

\medskip%

\end{sloppy}


\begin{thebibliography}{99.}%
\vspace{-2pt}%
\bibitem{caeirogomes_beirlant2012}%
Beirlant, J., Caeiro, F., Gomes, M.I.: An overview and open research topics in statistics of univariate extremes. Revstat \textbf{10}(1), 1-31 (2012)
%
\bibitem{caeirogomes_beirlant2004}%
Beirlant J., Goegebeur Y., Segers J., Teugels J.: Statistics of Extremes. Theory and Applications. Wiley (2004)
%
\bibitem{caeirogomes_beirlant2008}%
Beirlant J., Figueiredo F., Gomes M.I., Vandewalle B.: Improved Reduced-Bias Tail Index and Quantile Estimators. J. Statist. Plann. and Inference \textbf{138}(6), 1851-1870 (2008)
%
\bibitem{caeirogomes_caeiro2009}%
Caeiro, F., Gomes M. I., Henriques-Rodrigues, L.: Reduced-bias tail index estimators under a third order framework. Comm. Statist. Theory Methods \textbf{38}(7), 1019-1040 (2009)
%
\bibitem{caeirogomes_caeiro2011a}%
Caeiro, F., Gomes, M.I.: Semi-Parametric Tail Inference through Probability-Weighted Moments. J. Statist. Plann. Inference {\bf 141}, 937-950 (2011a)
\bibitem{caeirogomes_caeiro2011b}%
Caeiro, F., Gomes, M.I.: Asymptotic comparison at optimal levels of reduced-bias extreme value index estimators. Stat. Neerl. \textbf{65}, 462-488 (2011b)
%
\bibitem{caeirogomes_caeiro2011c}%
Caeiro, F., Gomes, M.I.: Computational validation of an adaptative choice of optimal sample fractions. Int. Statistical Inst.: Proc. 58th World Statistics Congress, Dublin 282-289 (2011c)
%
\bibitem{caeirogomes_caeiro2012}%
Caeiro, F. and Gomes, M.I. (2014). A semi-parametric estimator of a shape second order parameter. In Pacheco, A., Oliveira, M.R., Santos, R. and Paulino, C.D. (eds.), New Advances in Statistical Modeling and Application. 
Springer-Verlag, Berlin and Heidelberg,
in press.
%
\bibitem{caeirogomes_caeiro2013}%
Caeiro, F., Gomes, M.I.: A class of semi-parametric probability weighted moment estimators. In
Oliveira, P.E., da Gra\c{c}a Temido, M., Henriques, C. and Vichi, M. (Eds.), Recent Developments in Modeling and
Applications in Statistics, 139-147, Springer (2013)
%
\bibitem{caeirogomes_caeiroetal2012}%
Caeiro, F., Gomes, M.I., Vandewalle, B.: Semi-Parametric Probability-Weighted Moments Estimation Revisited. Accepted in Methodology and Computing in Applied Probability (2012) doi: 10.1007/s11009-
012-9295-6
%
%
\bibitem{caeirogomes_ciuperca2010}%
Ciuperca, G., Mercadier, C.: Semi-parametric estimation for heavy tailed distributions, Extremes , \textbf{13}(1), 55-87, (2010)
%
\bibitem{caeirogomes_david2003}%
David, H., Nagaraja, H.N.: {\it Order Statistics}. John Wiley \& Sons, New York (2003)
%
\bibitem{caeirogomes_dekkers2003}%
Dekkers, A., de Haan, L.: Optimal sample fraction in extreme value estimation. J. Multivariate Anal. \textbf{47}(2), 173-195 (1993)
%
%
\bibitem{caeirogomes_fragaalves2003}
Fraga Alves, M.I., Gomes, M.I.,  de Haan, L.:
 A new class of semi-parametric estimators of the second order parameter. Port. Math. \textbf{60}(2), 193-213 (2003)
%
\bibitem{caeirogomes_goegebeur2010}
Goegebeur, Y., Beirlant, J., de Wet, T.: Kernel estimators for the second order parameter in extreme value statistics. J. Statist. Plann. Inference \textbf{140}, 2632-2652 (2010)
%
\bibitem{caeirogomes_gomesmartins2001}%
Gomes, M. I. and M. J. Martins: Generalizations of the Hill estimator - asymptotic versus finite sample behaviour. J. Statist. Plann. Inference \textbf{93}, 161-180 (2001)
%
\bibitem{caeirogomes_gomesmartins2002}%
Gomes M.I., Martins M.J.: ``Asymptotically unbiased'' estimators of the tail index based on external estimation of the second order parameter. Extremes \textbf{5}(1), 5-31 (2002)

\bibitem{caeirogomes_gomesetal2008}%
Gomes, M.I., Canto e Castro, L., Fraga Alves, M.I., Pestana, D.D.: Statistics of extremes for IID data and breakthroughs in the estimation of the extreme value index: laurens de haan leading contributions. Extremes \textbf{11}(1), 3-34 (2008)
%
\bibitem{caeirogomes_greenwood1979} 
Greenwood, J. A, Landwehr, J. M., Matalas, N. C.,  Wallis, J.R.: Probability Weighted Moments: Definition and Relation to Parameters of Several Distributions Expressable in Inverse Form. Water Resources Research {\bf 15}, 1049-1054 (1979)
%
\bibitem{caeirogomes_dehaan2006}
de Haan, L., Ferreira, A.: Extreme Value Theory: An Introduction. Springer, New York (2006)
%
\bibitem{caeirogomes_dehaan1998}%
de Haan, L., Peng, L.:  Comparison of Tail Index Estimators. Stat. Neerl. {\bf 52}, 60-70 (1998)
%
\bibitem{caeirogomes_hill1975}%
Hill, B.M.:  A simple general approach to inference about the tail of a distribution. Ann. Statist. {\bf 3}, 1163-1174 (1975)
%
\bibitem{caeirogomes_hosking1987} 
Hosking, J.,   Wallis, J.:  Parameter and quantile estimation for the Generalized Pareto distribution. Technometrics {\bf 29}(3), 339-349 (1987)
%
%
\bibitem{caeirogomes_weissman1978} %
Weissman, I.:  Estimation of Parameters of Large Quantiles Based on the $k$ Largest Observations. J. Amer. Statist. Assoc. {\bf 73}, 812-815 (1978)
%
\end{thebibliography}
\end{document}